\documentclass[12pt,a4paper]{article}
\usepackage{amsmath,bm,amssymb,amsthm,mathrsfs,setspace}
\usepackage{slashbox}
\usepackage[dvipdfmx,bookmarksnumbered=true,pdfborder={0 0 0},]{hyperref}
\usepackage[top=30truemm,bottom=30truemm,left=25truemm,right=25truemm]{geometry}\usepackage{setspace} 
\usepackage[dvips]{graphicx}

\newtheorem{lemm}{Lemma}

\newtheorem{prop}{Proposition}

\title{Implications of the Tradeoff between Inside and Outside Social Status in Group Choice\thanks{I am grateful to Akihiko Matsui, Daisuke Oyama, Masayuki Yagasaki, Ryuichi Tanaka, participants of 2017 Asian meeting of the Econometric Society at Hong Kong Chinese University, participants of 2016 Spring Meeting of the Japanese Economic Association at Nagoya University for their valuable comments. This project is financially supported by JSPS Grand-in-Aid for JSPS Research Fellows (17J07819). The previous title of this paper is ``The Effect of Inside and Outside Reputations on Group Choice and Effort''. All remaining errors are my own.}}
\author{Takaaki Hamada \thanks{Faculty of Management and Administration, Shumei University, 1-1 Daigakucho, Yachiyo-shi, Chiba 276-0003, Japan. Email: t.hamada.econ@gmail.com}}
\date{\today}

\begin{document}
\maketitle
\begin{abstract}
{\small
We investigate a group choice problem of agents pursuing social status. We assume heterogeneous agents want to signal their private information (ability, income, patience, altruism, etc.) to others, facing tradeoff between ``outside status'' (desire to be perceived in prestigious group from outside observers) and ``inside status'' (desire to be perceived talented from peers inside their group). To analyze the tradeoff, we develop two stage signaling model in which each agent firstly chooses her group and secondly chooses her action in the group she chose. They face binary choice problems both in group and action choices. Using cutoff strategy, we construct an partially separating equilibrium such that there are four populations: (i) choosing high group with strong incentive for action in the group, (ii) high group with weak incentive, (iii) low group with strong incentive, and (iv) low group with weak incentive. By comparative statics results, we find some spillover effects from a certain group to another, on how four populations change, when a policy is taken in each group. These results have  rich implications for group choice problems like school, firm or residential preference.}
\end{abstract}

\newpage
\section{Introduction}
In our real world many groups are organized. Examples include friend circles, colleges, residential areas, sports teams, firms or criminal organizations like that. When people choose group, people are often concerned with two things;  ``outside status'' (desire to be perceived in prestigious group from outside observers) and ``inside status'' (desire to be perceived talented from peers inside their group). Take ``college choice" as example. Students may want to go to high respected college since they can be considered as having high ability by people outside the college. On the other hand, they may also prefer colleges where they are perceived as more excellent than other students by people inside the college. Almost all students confront the tradeoff since their abilities are limited.

In this paper we analyze effects of the tradeoff on agents' group choice and performance. We develop two stage signaling model in which each player firstly chooses her group, High or Low, and secondly chooses her action in the group, making an effort or not, with caring status payoffs. Each player has a private information, which is inferred through their group choice and action by spectators. We assume that while each player's group choice is observed by all other players, each player's effort level is observed only by people who are in the same group. 

By cutoff strategy, we construct an partially separating equilibrium such that four kinds of populations emerge: (i) choosing high group with strong incentive for action in the group, (ii) high group with weak incentive, (iii) low group with strong incentive, and (iv) low group with weak incentive. By comparative statics results, we find some spillover effects from a certain group to another, on how four populations change, when a policy is taken in each group. Let us consider one example. Suppose more agents in high group are motivated to make efforts by some policies. Then low-ability agents who get still lazy in high group become more miserable in the group since their inside status are getting low. As a result, they give up gaining the high outside status, and move to lower group. On the other hand, that makes low group members motivated since new comers' ability are relatively high in low group, and they get higher social status by pooling them. Our results can be applied to many real group choice problems, and also gives us new empirical implication for related phenomena.

In phycology and sociology, these two status concerns are considered to be important when people choose a group and their behaviors after joining a group. For example Tyler and Lind (1992) propose the group value model, where inside and outside status (in their paper, {\it pride} and {\it respect}) shape self-esteem and their behavior in groups. Heather et al. (1996), in their experiment, show both pride and respect are significantly related to self-esteem and their behavior in groups. Moreover, Branscombe et al. (2002) execute an experiment to show the importance of interaction between two components.

In economics these status concerns have been introduced to group formation model in which people have a preference for group members and it's argued that what kind of group are stable in some sense ( Milchtaich and Winter (2001), Watts (2007), Lararova and Dimitrov (2013)). Moreover, some papers discuss the effect of inside status on people's behavior in a given group (Benabou and Tirole (2006, 2011), Adriani and Sonderegger (2019)). None of these studies can incorporate both people's group choice and later performance, and my work is the first one which describe the interaction of group choice and later effort as long as I know.

This paper is organized as follows. Section $2$ describe our model, conducts equilibrium analysis and comparative statics, and Section $3$ presents applications and discuss the implication of our model. Section $4$ concludes.

\section{The Model}
There are two kinds of player; agents who choose a group and an action, and spectators who observe choices taken by agents. Agents are $[0,1]$ continuum and characterized by each type $\theta$ (ability, income, altruism, etc.), which is agents' private information. $\theta$ is distributed according to the density $f$: $[0,1]\rightarrow \mathbb{R_{+}}$ with mean $\overline{\theta}$. Agents firstly choose a group $g\in \{h, l\}$ they join, and then choose an action $a\in \{0,1\}$ in each group. In the example of school choice, $\theta$ is the ability for activities in school, and $h$ corresponds to high-ranking school, and $l$ corresponds to low-ranking school. Moreover, $a=1$ is to make an effort in school and $a=0$ is to get lazy. 

We assume each agent's preference depends on material utility and social status concerns. Using honor-stigma model in B\'{e}nabou and Tirole $(2006, 2011)$, our model is described as follows.
\begin{eqnarray}\label{gc}
\max_{g\in\{l, h\}} V(a^{*}, g, \theta)-c(g, \theta)+\mu_{O}\{\mathbb{E}[\theta|g]-\overline{\theta}\} \label{eq:a}
\end{eqnarray}
where
\begin{eqnarray}\label{ac}
V(a^{*}, g, \theta)\equiv \max_{a\in \{0, 1\}}v(a,g)-d(a, g, \theta)+\mu_{I}\{\mathbb{E}[\theta|a,g]-\overline{\theta}_g \}, \label{eq:b}
\end{eqnarray}
with $\overline{\theta}_g$ average type of group $g$. First, $(\ref{ac})$ is on action choice after joining a group $g$. Given a group $g$, $V(a^{*}, g, \theta)$ is a value function of type $\theta$'s action choice problem, and $a^{*}$ is the optimal action. $v(a,g)$ and $d(a, g, \theta)$ is a standard economic benefit and cost of choosing $a$ in a group $g$, respectively. The cost function is continuously differentiable, and satisfies the basic signaling properties; $\frac{\partial}{\partial \theta}d(a, g, \theta)<0$ for all $a$, $g$, $\theta$ and $\frac{\partial }{\partial \theta}d(1, g, \theta)<\frac{\partial }{\partial \theta}d(0, g, \theta)$ for all $g$, $\theta$.\footnote{The later one is rewritten as $\frac{\partial }{\partial \theta}[d(1, g, \theta)-d(0, g, \theta)]<0$.} The higher type agents are, the less cost they incur when to take an action. And the difference of cost between $a=1$ and $a=0$ is smaller in higher type.

The term $\mu_{I}\{\mathbb{E}[\theta|a,g]-\overline{\theta}_g \}$ is a psychological payoff from inside status, which depends on the difference between each agent's perceived type $\mathbb{E}[\theta|a,g]$ and average type in a group $g$. $\mu_I(>0)$ is agents' sensitivity to the inside status, and we assume it's common among agents.\footnote{B\'{e}nabou and Tirrole (2006, 2011) consider this parameter as ``identity" of group or community. It can be different from group to group. Moreover, they also think of it as the ``observability'', that is the degree of how easy spectators can observe agents' actions.} When $\mathbb{E}[\theta|a,g]>\overline{\theta}_g$ is positive, agents is subjected to ``honor'', and otherwise, ``stigma''. We can see this form expresses positional utility in a group. In fact, B\'{e}nabou and Tirrole (2006, 2011) do not use this formulation and their form is just $\mu_{I}\mathbb{E}[\theta|a,g]$, since $\overline{\theta}_g$ has no effect on agents' action choice. But in our two stage model, it critically matters to the group choice. If we except $\overline{\theta}_g$, it holds in an equilibrium that $\mu_{I}\{\mathbb{E}[\theta|a=1,g=l]<\mu_{I}\{\mathbb{E}[\theta|a=0,g=h]$ since higher types choose high-ranking group by signaling property. In this form, however, although agents choosing $g=h$ and $a=0$ are perceived relatively low type by group members, they get higher psychological payoffs on inside status than agents who exert relatively high performance $(a=1)$ in $l$ group. This does not describe the trade between inside and outside status. By using positional utility, $\mu_{I}\{\mathbb{E}[\theta|a=1,g=l]-\overline{\theta}_l\}>\mu_{I}\{\mathbb{E}[\theta|a=0,g=h]-\overline{\theta}_h\}$ can hold, from which the tradeoff arises.

Next, $(\ref{gc})$ is on group choice. $c(g, \theta)$ is the cost for joining group $g$ and, as in the cost function for action, it satisfies the signaling principal; $\frac{\partial }{\partial \theta}c(g, \theta)<0$ for any $g$, $\theta$ and $\frac{\partial }{\partial \theta}c(h, \theta)<\frac{\partial }{\partial \theta}c(l, \theta)$ for any $\theta$. $\mu_{O}\{\mathbb{E}[\theta|g]-\overline{\theta}\}$ is 
a psychological payoff from outside status, and $\mu_O(>0)$ is agents' sensitivity for outside status and is common among all agents. Note that outside status does not depend on an action $a$, while inside status does. We assume that spectators out of group cannot observe each agent's action in a group.\footnote{At the same time, each agent's inside status is formed by predictions of agents in a group. So agents can also be spectators.} In fact, this assumption does not change our equilibrium qualitatively, but this causes quantitative change. Our assumption implies the prediction of outside spectators become more rough, so that even around-bottom types in each group can get same outside status with types around top. This provides people with stronger incentive to join higher group, which results in more lazy actions in high group.

For the tractability, we do not consider group capacity constraint. On the other hands, it is essential for some phenomena such as school or firm choice. it should be discussed as a further analysis.

The sequence of game is as follows:

\begin{itemize}
\item[{\bf Stage 1}]: (i) Each agent's ability $\theta\in[0, 1]$ is realized and observed by themselves. \vspace{2mm} \\
 \;\;(ii) Each agent chooses group $g\in\{l, h\}$.\vspace{2mm} \\ 
 \;\;(iii) Spectators outside group observe agents' group choices, and update their belief on type for each agent.
\item[{\bf Stage 2}]:  (i) Each agent chooses action $a\in\{0, 1\}$ in the group $g$ chosen by her.\vspace{2mm} \\
 \;\;(ii) Spectators (agents) inside group observe agents' actions, and update their belief on type for each agent.\vspace{2mm} \\
 \end{itemize}
 
\subsection{Equilibrium Analysis}
An equilibrium concept is {\it perfect bayesian equilibrium}. In particular, as in B\'{e}nabou and Tirrole (2006, 2011) and Adriani and Sonderegger (2019), we focus on a cutoff equilibrium, where all agents' behaviors are characterized by a certain point. Let $\hat{\theta}$ be a group cutoff, where any types $\theta \geq \hat{\theta}$ join group $g$ and $\theta< \hat{\theta}$ join group $l$. Moreover, for each $g\in \{h, l\}$, let $\theta_{g}$ be a action cutoff, where any types $\theta \geq \theta_{g}$ take $a=1$ and $\theta< \theta_{g}$ take $a=0$ in each group. Figure \ref{ex} is an example, an we shall derive it as an equilibrium outcome.

\begin{figure}[h]
\begin{center}
\includegraphics[width=13cm,clip]{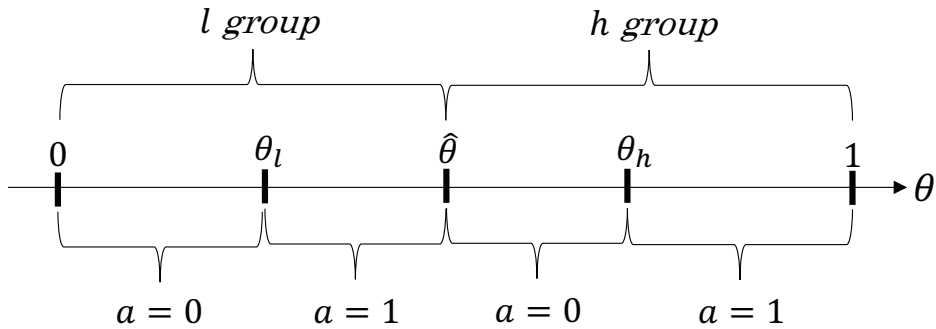}
\caption{An example of equilibrium outcome}\label{ex}
\end{center}
\end{figure}

\subsubsection{Optimal action choice}
 First backwardly, consider the problem $(\ref{eq:b})$ given a group cutoff $\hat{\theta}$. Optimal action of type $\theta$ in group $g$ is $a=1$ if
\begin{eqnarray}
v(1,g)-d(1, g, \theta)+\mu_I \mathbb{E}[\theta|a=1, g]\geq v(0,g)-d(0, g, \theta)+ \mu_I \mathbb{E}[\theta|a=0, g] \label{eq:c}
\end{eqnarray}

Let an equilibrium cutoff for action with a group cutoff $\hat{\theta}$ be $\theta_g^{*}$, and this satisfies
\begin{eqnarray}
d_{g}(\theta_g^{*})-v_{g}= \mu_I \phi_g(\theta_g^{*}), \label{eq:d}
\end{eqnarray}
where $v_{g}\equiv v(1,g)-v(0,g)$, $d_{g}(\theta_g)\equiv d(1, g, \theta_g)-d(0, g, \theta_g)$, and $\phi_g(\theta_g)\equiv \mathbb{E}[\theta|a=1, g]-\mathbb{E}[\theta|a=0, g]=\mathbb{E}_[\theta|\theta_g^{+}\geq \theta \geq \theta_g]-\mathbb{E}[\theta|\theta_g>\theta \geq \theta_g^{-}]$ with $\theta_g^{+}$ and $\theta_g^{-}$ being highest type and lowest type in group $g$, respectively.\footnote{$\theta_h^{+}=1$, $\theta_h^{-}=\theta_l^{+}=\hat{\theta}$, and $\theta_l^{-}=0$.}

\vspace{5mm}
 Let us consider the example that $f$ is the uniform distribution. Then we can calculate $\phi_h(\theta_h)=\frac{1-\hat{\theta}}{2}$ and $\phi_l(\theta_l)=\frac{\hat{\theta}}{2}$, which are constant in $\theta_{g}$. And consider $d_{h}(\theta_h)=\delta_{h}\cdot (1-\theta_{h})$ and $d_{l}(\theta_l)=\delta_{l}\cdot (1-\theta_{l})$ with $\delta_{h}>\delta_{l}>0$. Moreover, $v_{h}=v_{l}=0$. Then action cutoffs can be determined by Figure \ref{ex2}.
 
\begin{figure}[h]
\begin{center}
\includegraphics[width=17cm,clip]{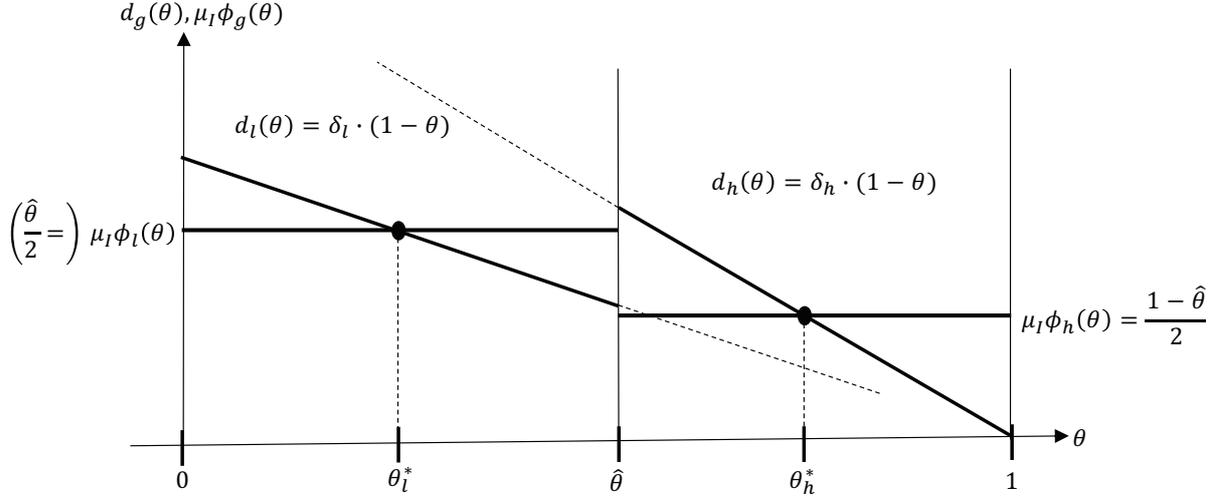}
\caption{A simple example of action cutoffs}\label{ex2}
\end{center}
\end{figure}

In this case, there exists an internal cutoff both in group $h$ and $l$. It's possible there is no internal cutoff, so that pooling equilibrium occurs in a group. Moreover in this case, if a cutoff exist, then it is unique. This comes from monotonicity of $d_{g}(\theta_g)$ and $\phi_g(\theta_g)$. 

\vspace{5mm}
In our assumption of cost function $d$, $d_g(\theta_g)$ is strictly decreasing. And if we consider the (weakly) monotone density $f$, then $\phi_g(\theta_g)$ is also (weakly) monotone as the next lemma shows.

\begin{lemm}$(Jewitt, 2004)$ If the density $f$ is $(a)$non-increasing, $(b)$non-decreasing, $(c)$constant, then $z(\tilde{\theta})\equiv \mathbb{E}[\theta|\theta^{+}\geq \theta \geq \tilde{\theta}]-\mathbb{E}[\theta|\tilde{\theta} \geq \theta \geq \theta^{-}]$ is $(a)$non-decreasing, $(b)$non-increasing, $(c)$constant in $\tilde{\theta} \in [\theta^{-},\theta^{+}]$.
\end{lemm}

Thus, if we consider the monotone density, then an equilibrium action cutoff is at most $1$ in each group. In our situation, it may be possible to consider $f$ is decreasing function. The distributions of ability or income are decreasing in most part.

\subsubsection{Optimal group choice}

Since $\theta_g^*$ depends on $\hat{\theta}$, let us denote it as $\theta_g^*(\hat{\theta})$.
Now consider the stage 1. By $(\ref{eq:a})$, an optimal group of type $\theta$ is $g=h$ if
\begin{eqnarray}
V(a^*(\theta_h^*(\hat{\theta})), h, \theta)-c(h, \theta)+\mu_{O}\mathbb{E}[\theta|h]\geq V(a^*(\theta_l^*(\hat{\theta})), l, \theta)-c(l, \theta)+\mu_{O}\mathbb{E}[\theta|l] \label{eq:e}
\end{eqnarray}
where $a^*(\theta_g^*(\hat{\theta}))$ is type $\theta$'s optimal action when she is in a group $g$ and group cutoff is $\hat{\theta}$. Thus an equilibrium cutoff $\hat{\theta}^*$ satisfies
\begin{eqnarray}
\mu_{O}\phi(\hat{\theta}^*)=\tilde{c}(\hat{\theta}^*)-\tilde{V}(\hat{\theta}^*)\label{eq:f}
\end{eqnarray}
where $\tilde{c}(\hat{\theta})\equiv c(h, \hat{\theta})-c(l, \hat{\theta})$, $\tilde{V}(\hat{\theta})\equiv V(a^*(\theta_h(\hat{\theta})), h, \theta)- V(a^*(\theta_l(\hat{\theta}))$, and $\phi(\hat{\theta})\equiv \mathbb{E}[\theta|h]-\mathbb{E}[\theta|l]=\mathbb{E}[\theta|1\geq \theta \geq \hat{\theta}]-\mathbb{E}[\theta|\hat{\theta}\geq \theta \geq 0]$. By the assumption of cost function $c$ and Lemma $1$, as in the analysis of action cutoff, we can see an equilibrium group cutoff is at most $1$ if the density is a monotone function.

\vspace{3mm}
We use belief refinement by applying Cho and Kreps (1987)'s criterion to exclude the unrealistic pooling equilibria. Moreover, for comparative statics analysis, we need to guarantee the stability of our equilibrium. If an equilibrium is not stable, an equilibrium cutoff diverges by deviations  and comparative statics are nonsense. The details are discussed in Appendix. We have the following proposition.

\begin{prop}$($Existence, Uniqueness and Stability of partially-separating equilibrium$)$\\
There exists a partially-separating equilibrium such that $\theta_h^*(\hat{\theta}^*)$, $\hat{\theta}^*$ and $\theta_l^*(\hat{\theta}^*)$ exist in $[0,1]$ and $1>\theta_h^*(\hat{\theta}^*)>\hat{\theta}^*>\theta_l^*(\hat{\theta}^*)>0$ holds, that is\vspace{3mm} \\
$\bullet$Type $\theta\in [0, \theta_l^*(\hat{\theta}^*)]$ join $l$ and choose $a=0$, \vspace{2mm} \\
$\bullet$Type $\theta\in [\theta_l^*(\hat{\theta}^*), \hat{\theta}^*]$ join $l$ and choose $a=1$,\vspace{2mm} \\
$\bullet$Type $\theta\in [\hat{\theta}^*, \theta_h^*(\hat{\theta}^*)]$ join $h$ and choose $a=0$,\vspace{2mm} \\
$\bullet$Type $\theta\in [\theta_h^*(\hat{\theta}^*), 1]$ join $h$ and choose $a=1$. \vspace{3mm}\\
 If this equilibrium exists, then the equilibrium outcome is unique with belief refinement D1. Moreover, an equilibrium is stable if $d_{g}'(\theta_g^*)-\mu_I\phi_g'(\theta_g^*)<0 $ for each $g \in \{h,l\}$ and $\tilde{c}'(\hat{\theta}^*)-\tilde{V}'(\hat{\theta}^*)-\mu_O\phi'(\hat{\theta}^*)<0$.
\end{prop}
\begin{proof}
See Appendix A.
\end{proof}
Proposition $1$ says that there exists an equilibrium such that there are agents with positive measure in all segments $(g, a)$. 
Before the next analysis, let us highlight each status concern. We can see inside status provides incentive to take $a=1$ because agents care the relative rank in their group. Now let's consider the agents choosing $g=h$ and $a=0$. They incur the psychological cost (stigma) from inside status since they get lazy ($a=0$) in group $h$, which signals their ability are relatively low in group $h$. The reason why they choose group $h$ in spite of the stigma is that they can be perceived to be talented by outside spectators, with putting economic payoff aside. 
The desire for outside status may make people pursue higher group but get lazy in a group.

In the rest of this paper, we focus on the equilibrium outcome in proposition $1$, and use comparative statics to analyze the tradeoff of two status concerns.

\subsection{Comparative Statics}
 First we consider the change of $d_{h}(\theta)$. 
\begin{prop}$($Comparative Statics $1$$)$\\
Consider the equilibrium in proposition $1$. If $d_{h}(\theta)$ increases equally for all $\theta$, then
\begin{itemize}
\item[$(i)$]$\hat{\theta}$ decreases,
\item[$(ii)$]$\theta_h$ increases, and
\item[$(iii)$]$\theta_l$ increases.
\end{itemize}
\end{prop}
\begin{proof}
See Appendix C.
\end{proof}

The following figure describe the result of proposition $2$.
\begin{figure}[h]
\begin{center}
  \includegraphics[width=10cm,clip]{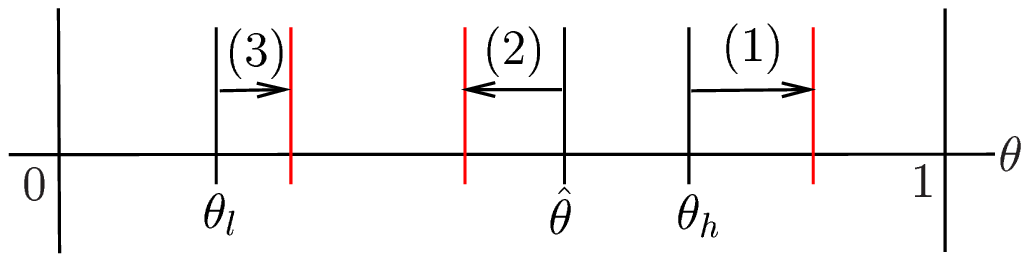}
\end{center}
\end{figure}

Since the increase for $d_h(\theta)$ means the cost of $a=1$ in $h$ group relatively goes up, more people in $h$ group get to choose $a=0$ $((1))$. Then the perceived type choosing $a=0$ in $h$ group becomes higher than ever because more people who has higher type get to choose $a=0$. This means the stigma to get lazy in $h$ decreases, so that more people in $l$ group move to $h$ group $((2))$ to attain higher outside status in $h$ group, giving up gaining high inside status in $l$ group. This means the honor to make an effort in $l$ decreases, so that more people in $l$ get lazy ((3)). 

It's interesting to consider the inverse direction. If $d_{h}(\theta)$ decreases equally for all $\theta$, then, by proposition 2, $(i)$ $\hat{\theta}$ increases, $(ii)$ $\theta_h$ decreases, and $(iii)$ $\theta_l$ decreases. Then we can increase total efforts ``without'' losing any types' effort as Figure \ref{policy_1} shows.

\begin{figure}[h]
\begin{center}
  \includegraphics[width=13cm,clip]{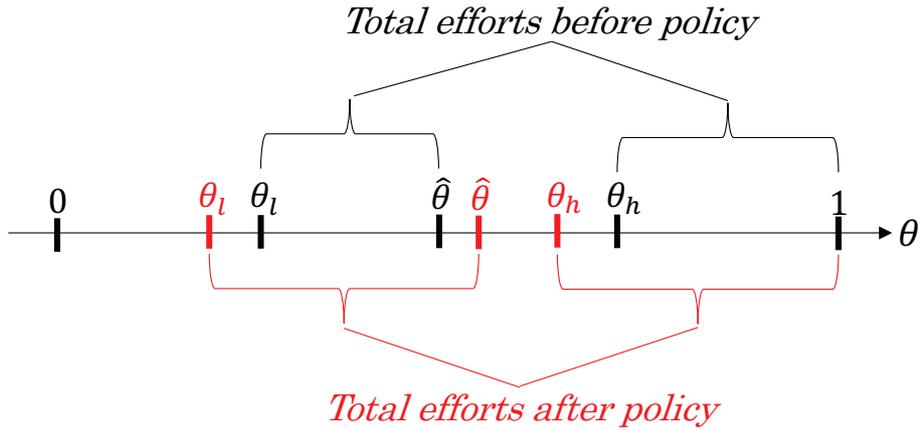}
\caption{Increase of $d_{h}(\theta)$ and total efforts}\label{policy_1}
\end{center}
\end{figure}

The decreases of action cost in $h$ group incentivize more agents in $h$, which makes some around-bottom agents more miserable, so that some lazy agents move to $l$ group. That has a positive effect on $l$ group since the honor to make an effort in $l$ group increases, which incentivize more agents in $l$.

 This provide us with an important intuition such that It's to incentivize higher group that has a positive spillover to lower group with respect to total efforts. This can be realized by increasing $\mu_{I}$ in $h$ group.\footnote{This is not explicitly analyzed in our model.} One potential way is to increase observability in $h$ group. 

 Next we consider the marginal change of $d_{l}(\theta)$. 

\begin{prop}$($Comparative Statics $2$$)$\\
Consider the equilibrium in proposition 1. If $d_{l}(\theta)$ increases equally for all $\theta$, then
\begin{itemize}
\item[$(i)$]$\hat{\theta}$ increases,
\item[$(ii)$]$\theta_h$ increases, and
\item[$(iii)$]$\theta_l$ increases.
\end{itemize}
\end{prop}
\begin{proof}
See Appendix C.
\end{proof}

The following figure describe the result of proposition $3$.

\begin{figure}[h]
\begin{center}
\includegraphics[width=10cm,clip]{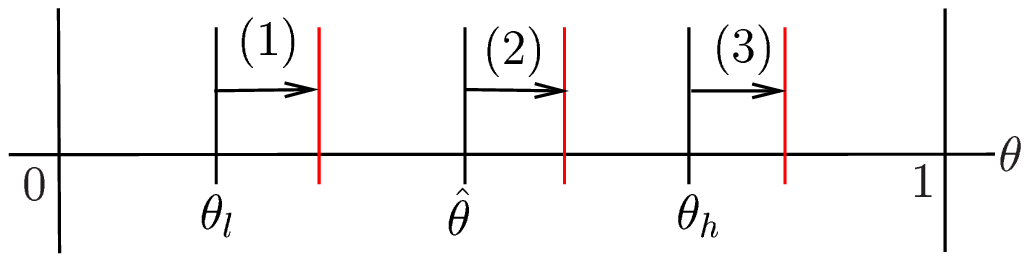}
\end{center}
\end{figure}

When the cost making an effort in $l$ group relatively goes up, more agents in $l$ group become lazy $((1))$. Then the perceived type choosing $a=1$ in $l$ becomes higher than ever because more agents having lower ability get to choose $a=0$. This means the honor to make an effort in $l$ increases, so that more people in $h$ group move to $l$ group $((2))$ to attain higher inside status in $l$ group, giving up gaining high outside status in $h$ group. On the other hand, in $h$ group, since some around-bottom agents move to $l$ group, the perceived ability given $a=0$ in $h$ becomes higher. This means the stigma not to make an effort in $h$ decreases, so that more people in $h$ get lazy ((3)). In this case, we cannot increase total efforts without losing any types' effort as proposition $2$.

 Finally we consider the marginal change of $\tilde{c}(\theta)$. 

\begin{prop}$($Comparative Statics $3$$)$\\
Consider the equilibrium in proposition 1. If $\tilde{c}(\theta)$ increases equally for all $\theta$, then
\begin{itemize}
\item[$(i)$]$\hat{\theta}$ increases,
\item[$(ii)$]$\theta_h$ increases, and
\item[$(iii)$]$\theta_l$ decreases.
\end{itemize}
\end{prop}
\begin{proof}
See Appendix C.
\end{proof}

The following figure describe the result of proposition $4$.

\begin{figure}[h]
\begin{center}
\includegraphics[width=10cm,clip]{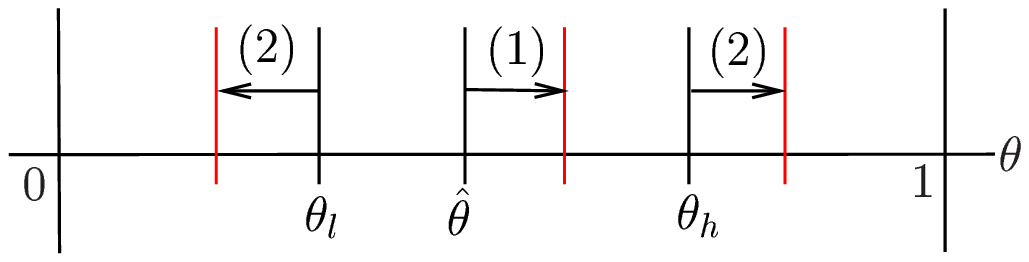}
\end{center}
\end{figure}

When the cost of entering group $h$ relatively goes up, more agents in $h$ get to move $l$ group $((1))$. Then there are two effects. First, the perceived type choosing $a=1$ in group $l$ becomes higher because more people having higher type get to choose $a=1$ than ever. This means the honor to make an effort in $l$ increases, so that more people get to choose to make an effort $((2))$. At the same time, the perceived type choosing $a=0$ in $h$ becomes higher because around-bottom types in $h$ move to $l$ group. This means the stigma to get lazy in $h$ decreases, so that more agents get lazy in $h$ $((3))$. Also in this case, we cannot increase total efforts without losing any types' effort.

\section{Applications}
\subsection{Job Market Signaling and Human Capital Accumulation}
It is said that many college students have studied less recently. For example, as Babcock and Marks (2011) show empirically there has been a decline in the amount of student's study effort as competition of the entrance exam has been getting severe. While many theoretical models, for example Bentley Marleon and Miguel Urquiola (2015) or Hamada (2016), try to explain this phenomenon. From the view point of the tradeoff between inside and outside status, we may explain as follows. Suppose as the competition of the entrance exam has been getting severe, students focus on ``entering college'', and become careless for inside status than ever.\footnote{Our model cannot explain this part.} In particular, if $\mu_{I}$ decreases in $h$ group, the stigma decreases of being lazy in $h$ college. Then, more students in $l$ college move to $h$ college to attain outside status, resulting in less total efforts as in the proposition $2$.

We can consider the policy that promotes more students to make efforts after entering college. That is ``to promote high level college students to make an effort''. Then the stigma to get lazy in $h$ college increases, so that some bottom students in $h$ college move to $l$ college. Then for $l$ college, since relatively high type students enter, the honor to make an effort increases. Then more students in $l$ college get to make an effort. The policy in high level college can be positive spillover to the lower college. Actually, it's natural to interrupt this results as the cross-years transition of social norm and students' behaviors rather than the instant changes.


\subsection{Status-seeking in criminal subculture}
Consider criminal subculture where people receive status benefit from revealing their toughness as Dur and Weele $(2012)$ assume. Now suppose there are two organizations High and Low, people choose which group they enter. After entering, they choose whether to commit crime. 
In this situation, to increase $d_g(\theta)$ corresponds to increase punishment of committing crime.

What is the implication of our model in this context? For example, If we increase punishment in $h$ group, then crime decreases in $h$ and stigma not to commit crime decreases. Then, it becomes easier for $l$ group member to enter the $h$ group, and this is how $h$ group expand far from shrinking. However, if the top of $l$ members move to $h$, then the honor to commit a crime in $l$ decreases. This make some agents quit to commit crime in $l$ group. Thus the increase in $d_h(\theta)$ decreases the whole crime rate. 

On the other hand, if we increases punishment in $l$ group, then the crime rate in $l$ decreases. However, this increases honor to commit a crime in $l$ and thus the bottom of $h$ members move to $l$. This makes stigma not to commit a crime decrease, and thus the crime in $h$ decreases. In this case, the crime rate in $h$ decreases and $h$ group shrink, but that in $l$ may increases.

\subsection{ Residence Choice and Conspicuous Consumption}
Consider the people who want to signal their own wealth or income. In literature of conspicuous consumption, each individual use visible goods as a signaling device such as rich car, watch or clothes like that. In this context, the resident must be one of useful devise. For example, recent empirical study shows that some people move from rural area to the urban by signaling their status (Janabel (1996), Sivanathan and Pettit (2010)). 

 Now suppose people choose a residential area $(g)$ and whether to consume a visible good $(a)$. People gain a psychological benefit from living rich area (outside status) or having expensive visible good (inside status). We can consider the following problem by using my model: When the planner's objective is to decrease consumptions of visible good in order to mitigate the harmful consumption race, what kind of policy is effective? The answer from my model is ``to increase the tax of visible good consumed in $h$ area''.

\section{Conclusion}
We have analyzed agents who have psychological benefits from both outside status (desire to be perceived in prestigious group from outside observers) and inside status (desire to be perceived talented from peers inside their group). Using signaling game, we develop two stage model to analyze the tradeoff between two social status concerns on people's group choice and action. We get useful comparative statics results, especially we find some spillover effects from a certain group to another, when a policy is taken in each group, through which many phenomena are explained and rich policy implications are driven. In our model, we consider only two groups, but we may extend our main result to more-than-three groups' setting. Even in these setting, it's predicted that to incentivize agents in the highest group has a positive spillover effects to all lower groups. 

On the other hands, our results depend on no capacity constraint. We should do the further analysis on how our results are preserved under constraint. Moreover, we cannot understand in this model why people care social status. We want to explore the origin of two status concerns, and derive the desire endogenously. This is a very big issue.


\section*{Appendix A: Preliminary}
First we prove two lemmas.
\begin{lemm}
For any differentiable full-support density $f$ it holds $\frac{\partial \phi_h(\hat{\theta})}{\partial \hat{\theta}}<0$ for all $\theta_h\in (\hat{\theta}, 1]$ and $\frac{\partial \phi_l(\hat{\theta})}{\partial \hat{\theta}}>0$ for all $\theta_l\in [0, \hat{\theta})$.
\end{lemm}

\begin{proof}
For any group threshold $\hat{\theta}\in (0, 1)$, we define two density functions $f_h$ and $f_l$ which represent each group density as follows:
\begin{eqnarray} 
\nonumber
f_h(\theta)&\equiv& \frac{f(\theta)}{1-F(\hat{\theta})}\;\;\;\;\theta \in [\hat{\theta}, 1],  \\ \nonumber
f_l(\theta)&\equiv& \frac{f(\theta)}{F(\hat{\theta})}\;\;\;\;\theta \in [0, \hat{\theta}],
\end{eqnarray}
where $f$ is the whole density function of players. We can calculate $\phi_g(\hat{\theta})$ as follows:
\begin{eqnarray}
\phi_h(\hat{\theta})&=&\frac{\int_{\theta_h}^{1}\theta dF}{1-F(\theta_h)}-\frac{\int_{\hat{\theta}}^{\theta_h}\theta dF}{F(\theta_h)-F(\hat{\theta})}, \\ 
\phi_l(\hat{\theta})&=& \frac{\int_{\theta_l}^{\hat{\theta}}\theta dF}{F(\hat{\theta})-F(\theta_l)}-\frac{\int_{0}^{\theta_l}\theta dF}{F(\theta_l)}.
\end{eqnarray}
Taking each derivative with respect to $\hat{\theta}$, we can have
\begin{eqnarray} 
\nonumber
\frac{\partial \phi_h(\hat{\theta})}{\partial \hat{\theta}}&=& -\frac{-\hat{\theta}f(\hat{\theta})\{F(\theta_h)-F(\hat{\theta})\}+f(\hat{\theta})\int_{\hat{\theta}}^{\theta_h}\theta dF}{\{F(\theta_h)-F(\hat{\theta})\}^2} \\ \nonumber
&=& \frac{f(\hat{\theta})\{\int_{\hat{\theta}}^{\theta_h}\hat{\theta} dF-\int_{\hat{\theta}}^{\theta_h}\theta dF\}}{\{F(\theta_h)-F(\hat{\theta})\}^2} \\
&=& \frac{f(\hat{\theta})\int_{\hat{\theta}}^{\theta_h}(\hat{\theta}-\theta) dF}{\{F(\theta_h)-F(\hat{\theta})\}^2}<0,\;\forall \theta_h \in (\hat{\theta}, 1]
\end{eqnarray}

and 

\begin{eqnarray} 
\nonumber
\frac{\partial \phi_l(\hat{\theta})}{\partial \hat{\theta}}&=& \frac{\hat{\theta}f(\hat{\theta})\{F(\hat{\theta})-F(\theta_l)\}-f(\hat{\theta})\int_{\theta_l}^{\hat{\theta}}\theta dF}{\{F(\hat{\theta})-F(\theta_l)\}^2} \\ \nonumber
&=&\frac{f(\hat{\theta})\{\int_{\theta_l}^{\hat{\theta}}\hat{\theta} dF-\int_{\theta_l}^{\hat{\theta}}\theta dF\}}{\{F(\hat{\theta})-F(\theta_l)\}^2} \\ 
&=&\frac{f(\hat{\theta})\int_{\theta_l}^{\hat{\theta}}(\hat{\theta}-\theta) dF}{\{F(\hat{\theta})-F(\theta_l)\}^2}>0, \;\forall \theta_l\in [0, \hat{\theta}).
\end{eqnarray}
\end{proof}

\begin{lemm}
For any differentiable full-support density $f$ followings hold:
\begin{eqnarray}
\nonumber
(i)\;\; \frac{\partial}{\partial \theta_h} \frac{\int_{\hat{\theta}}^{\theta_h}\theta f(\theta)d\theta}{F(\theta_h)-F(\hat{\theta})}>0,\\ \nonumber
(ii)\;\; \frac{\partial}{\partial \theta_l} \frac{\int_{\theta_l}^{\hat{\theta}}\theta f(\theta)d\theta}{F(\hat{\theta})-F(\theta_l)}<0.
\end{eqnarray}
\end{lemm}
\begin{proof}
(i) : By the fundamental theorem of calculus, we have
\begin{eqnarray}
\nonumber
\frac{\partial}{\partial \theta_h} \frac{\int_{\hat{\theta}}^{\theta_h}\theta f(\theta)d\theta}{F(\theta_h)-F(\hat{\theta})}&=&\left[-\frac{f(\theta_h)}{(F(\theta_h)-F(\hat{\theta}))^2}\int_{\hat{\theta}}^{\theta_h}\theta f(\theta)d\theta+\frac{\theta_h f(\theta_h)}{F(\theta_h)-F(\hat{\theta})} \right]\\ \nonumber
&=&\left[\frac{f(\theta_h)}{F(\theta_h)-F(\hat{\theta})}\left(\theta_h-\int_{\hat{\theta}}^{\theta_h}\frac{\theta f(\theta)}{F(\theta_h)-F(\hat{\theta})}d\theta\right)\right]\\ \nonumber
&=&\left[\frac{f(\theta_h)}{F(\theta_h)-F(\hat{\theta})}\left(\theta_h-\mathbb{E}[\theta|\hat{\theta}<\theta<\theta_h]\right)\right]>0.
\end{eqnarray}
(ii) : Similarly to (i),
\begin{eqnarray}
\nonumber
\frac{\partial}{\partial \theta_l} \frac{\int_{\theta_l}^{\hat{\theta}}\theta f(\theta)d\theta}{F(\hat{\theta})-F(\theta_l)}&=&\left[\frac{f(\theta_l)}{(F(\hat{\theta})-F(\theta_l))^2}\int_{\theta_l}^{\hat{\theta}}\theta f(\theta)d\theta-\frac{\theta_l f(\theta_l)}{F(\hat{\theta})-F(\theta_l)} \right]\\ \nonumber
&=&\left[\frac{f(\theta_l)}{F(\hat{\theta})-F(\theta_l)}\left(\int_{\theta_l}^{\hat{\theta}}\frac{\theta f(\theta)}{F(\hat{\theta})-F(\theta_l)}d\theta -\theta_l\right)\right]\\ \nonumber
&=&\left[\frac{f(\theta_l)}{F(\theta_l)-F(\hat{\theta})}\left(\mathbb{E}[\theta|\theta_l<\theta<\hat{\theta}] -\theta_l\right)\right]<0.
\end{eqnarray}
\end{proof}

 \section*{Appendix B: Proof of Proposition 1}
(Being revised and updated soon) 

 \section*{Appendix C: Proof of Proposition 2, 3, 4}
 
\begin{proof}[Proof of Proposition 2:\nopunct]
Fix a partially-separating equilibrium characterized by $(\hat{\theta}^*, \theta_h^*, \theta_l^*)$. Let the increase of cost for an effort in $h$ be $\alpha$; that is each type $\theta$ making an effort in $h$ incurs the cost $d_{h}(\theta)+\alpha$. Now define the following functions:
\begin{eqnarray}
\sigma^h(\hat{\theta}, \theta_h, \theta_l, \alpha)&\equiv&d_{h}(\theta_h)+\alpha-v_{h}-\mu_{I}\phi_h(\hat{\theta})\\
\sigma^l(\hat{\theta}, \theta_h, \theta_l, \alpha)&\equiv&d_{l}(\theta_l)-v_{l}-\mu_{I} \phi_l(\hat{\theta})\\
\sigma^s(\hat{\theta}, \theta_h, \theta_l, \alpha)&\equiv&\mu_{I} \tilde{\phi}(\hat{\theta},\theta_h, \theta_l)-(\mu_I-\mu_O)\phi(\hat{\theta})-\tilde{d}(\hat{\theta})-\tilde{c}(\hat{\theta})
\end{eqnarray}
where $\tilde{\phi}(\hat{\theta},\theta_h, \theta_l)=\left[\frac{\int_{\hat{\theta}}^{\theta_h}\theta f(\theta)d\theta}{F(\theta_h)-F(\hat{\theta})}-\frac{\int_{\theta_l}^{\hat{\theta}}\theta f(\theta)d\theta}{F(\hat{\theta})-F(\theta_l)}\right]$, $\tilde{d}(\hat{\theta})=d(1, h, \hat{\theta})-d(0, l, \hat{\theta})$. Note that, since $(\hat{\theta}^*, \theta_h^*, \theta_l^*)$ satisfies $(\ref{eq:d})$ and $(\ref{eq:f})$ in the equilibrium, 
\begin{eqnarray}
\sigma^x(\hat{\theta}^*, \theta_h^*, \theta_l^*, 0)=0
\end{eqnarray}
holds for each $x\in \{h, l, s\}$. Calculating the derivative of each $\sigma^x$ with respect to $\hat{\theta}$, $\theta_h$ and $\theta_l$,
\begin{eqnarray}
(\sigma^h_h, \sigma^h_l, \sigma^h_s)&=&\left(\frac{\partial d_{h}(\theta_h)}{\partial \theta_h}-\mu_{I} \frac{\partial \phi_h(\hat{\theta})}{\partial \theta_h},\;\;0,\;\;-\mu_{I}\frac{\partial \phi_h(\hat{\theta})}{\partial \hat{\theta}}\right)\\
(\sigma^l_h, \sigma^l_l, \sigma^l_s)&=&\left(0,\;\;\frac{\partial d_{l}(\theta_l)}{\partial \theta_l}-\mu_{I} \frac{\partial \phi_l(\hat{\theta})}{\partial \theta_l},\;\;-\mu_{I} \frac{\partial \phi_l(\hat{\theta})}{\partial \hat{\theta}}\right)\\ \nonumber
(\sigma^s_h, \sigma^s_l, \sigma^s_s)&=&\Biggl( \mu_{I} \frac{\partial \tilde{\phi}(\hat{\theta},\theta_h, \theta_l)}{\partial \theta_h},\;\; \mu_{I}\frac{\partial \tilde{\phi}(\hat{\theta},\theta_h, \theta_l)}{\partial \theta_l},\\ 
&&\mu_I \frac{\partial \tilde{\phi}(\hat{\theta},\theta_h, \theta_l)}{\partial \hat{\theta}}-(\mu_I-\mu_O)\frac{\partial \phi(\hat{\theta})}{\partial \hat{\theta}}+\frac{\partial \tilde{d}(\hat{\theta})}{\partial \hat{\theta}}-\frac{\partial c(\hat{\theta})}{\partial \hat{\theta}}\Biggl)
\end{eqnarray}
where $\sigma_x^y=\frac{\partial \sigma^x}{\partial \theta_y}$ for each $x,y\in \{h, l, s\}$ and $\theta_s=\hat{\theta}$. By the stability of the equilibrium and lemma 2,3, we have in the followings in the equilibrium:
\begin{eqnarray}
\sigma_h^h<0,\;\;\sigma_l^h=0,\;\;\sigma_s^h>0,\;\;\sigma_h^l=0,\;\;\sigma_l^l<0,\;\;\sigma_s^l<0,\;\;\sigma_h^s>0,\;\;\sigma_l^s<0,\;\;\sigma_s^s>0.
\end{eqnarray}
Moreover, calculating the Jacobian determinant of $\sigma$,
\begin{eqnarray}
|J_\sigma|=
\begin{vmatrix}
\sigma_h^h &\sigma_l^h &\sigma_s^h\\
\sigma_h^l &\sigma_l^l &\sigma_s^l\\
\sigma_h^s &\sigma_l^s &\sigma_s^s
\end{vmatrix}
=\sigma_h^h(\sigma_l^l \sigma_s^s-\sigma_l^s \sigma_s^l)-\sigma_h^s \sigma_l^l \sigma_s^h>0.
\end{eqnarray}
By $(10)$ and $(16)$, we can apply implicit function theorem. Since $(\sigma^h_{\alpha}, \sigma^l_{\alpha}, \sigma^s_{\alpha})=(1,0,0)$, we can have in the equilibrium 
\begin{eqnarray}
\nonumber
\frac{\partial \hat{\theta}}{\partial \alpha}|_{\alpha=0} &=&-
\begin{vmatrix}
\sigma_h^h &\sigma_l^h &1\\
\sigma_h^l &\sigma_l^l &0\\
\sigma_h^s &\sigma_l^s &0
\end{vmatrix}
/|J_\sigma|=-(\sigma_h^l \sigma_l^s-\sigma_h^s \sigma_l^l)/|J_\sigma|<0, \\ \nonumber
\frac{\partial \theta_h}{\partial \alpha}|_{\alpha=0} &=&-
\begin{vmatrix}
1 &\sigma_l^h &\sigma_s^h\\
0 &\sigma_l^l &\sigma_s^l\\
0 &\sigma_l^s &\sigma_s^s
\end{vmatrix}
/|J_\sigma|=-(\sigma_l^l \sigma_s^s-\sigma_l^s \sigma_s^l)/|J_\sigma|>0, \\ \nonumber
\frac{\partial \theta_l}{\partial \alpha}|_{\alpha=0} &=&-
\begin{vmatrix}
\sigma_h^h &1 &\sigma_s^h\\
\sigma_h^l &0 &\sigma_s^l\\
\sigma_h^s &0 &\sigma_s^s
\end{vmatrix}
/|J_\sigma|=-(\sigma_h^s \sigma_s^l-\sigma_h^l \sigma_s^s)/|J_\sigma|>0.
\end{eqnarray}
\end{proof}

\begin{proof}[Proof of Proposition 3:\nopunct]
Fix a partially-separating equilibrium characterized by $(\hat{\theta}^*, \theta_h^*, \theta_l^*)$. Let the increase of cost for an effort in $l$ be $\beta$; that is each type $\theta$ making an effort in $l$ incurs the cost $d_{l}(\theta)+\beta$. Now as in the proof of prop 2 define the following functions:
\begin{eqnarray}
\sigma^h(\hat{\theta}, \theta_h, \theta_l, \alpha)&\equiv&d_{h}(\theta_h)-v_{h}-\mu_{I}\phi_h(\hat{\theta})\\
\sigma^l(\hat{\theta}, \theta_h, \theta_l, \alpha)&\equiv&d_{l}(\theta_l)+\beta-v_{l}-\mu_{I} \phi_l(\hat{\theta})\\
\sigma^s(\hat{\theta}, \theta_h, \theta_l, \alpha)&\equiv&\mu_{I} \tilde{\phi}(\hat{\theta},\theta_h, \theta_l)-(\mu_I-\mu_O)\phi(\hat{\theta})-\tilde{d}(\hat{\theta})-\tilde{c}(\hat{\theta})
\end{eqnarray}
Since $(\sigma_\beta^h, \sigma_\beta^l, \sigma_\beta^s)=(0,1,0)$ and the others' derivatives $\sigma_x^y$ is same as prop 2, by implicit function theorem, we have
\begin{eqnarray}
\nonumber
\frac{\partial \hat{\theta}}{\partial \beta}|_{\beta=0} &=&-
\begin{vmatrix}
\sigma_h^h &\sigma_l^h &0\\
\sigma_h^l &\sigma_l^l &1\\
\sigma_h^s &\sigma_l^s &0
\end{vmatrix}
/|J_\sigma|=-(\sigma_h^s \sigma_l^h-\sigma_h^h \sigma_l^s)/|J_\sigma|>0, \\ \nonumber
\frac{\partial \theta_h}{\partial \beta}|_{\beta=0} &=&-
\begin{vmatrix}
0 &\sigma_l^h &\sigma_s^h\\
1 &\sigma_l^l &\sigma_s^l\\
0 &\sigma_l^s &\sigma_s^s
\end{vmatrix}
/|J_\sigma|=-(\sigma_l^s \sigma_s^h-\sigma_l^h \sigma_s^s)/|J_\sigma|>0, \\ \nonumber
\frac{\partial \theta_l}{\partial \beta}|_{\ \beta=0} &=&-
\begin{vmatrix}
\sigma_h^h &0 &\sigma_s^h\\
\sigma_h^l &1 &\sigma_s^l\\
\sigma_h^s &0 &\sigma_s^s
\end{vmatrix}
/|J_\sigma|=-(\sigma_h^h \sigma_s^s-\sigma_h^s \sigma_s^h)/|J_\sigma|>0.
\end{eqnarray}

\end{proof}

\begin{proof}[Proof of Proposition 4:\nopunct]
Fix a partially-separating equilibrium characterized by $(\hat{\theta}^*, \theta_h^*, \theta_l^*)$. Let the increase of cost for entering $h$ be $\gamma$; that is each type $\theta$ joining $h$ incurs the cost $\tilde{c}(\theta)+\gamma$. Now as in the proof of prop 2,3 define the following functions:
\begin{eqnarray}
\sigma^h(\hat{\theta}, \theta_h, \theta_l, \alpha)&\equiv&d_{h}(\theta_h)-v_{h}-\mu_{I}\phi_h(\hat{\theta})\\
\sigma^l(\hat{\theta}, \theta_h, \theta_l, \alpha)&\equiv&d_{l}(\theta_l)-v_{l}-\mu_{I}\phi_l(\hat{\theta})\\
\sigma^s(\hat{\theta}, \theta_h, \theta_l, \alpha)&\equiv&\mu_{I}\tilde{\phi}(\hat{\theta},\theta_h, \theta_l)-(\mu_I-\mu_O)\phi(\hat{\theta})-\tilde{d}(\hat{\theta})-\tilde{c}(\hat{\theta})-\gamma
\end{eqnarray}
Since $(\sigma_\gamma^h, \sigma_\gamma^l, \sigma_\gamma^s)=(0,0,-1)$ and the others' derivatives $\sigma_x^y$ is same as prop 2, by implicit function theorem, we have
\begin{eqnarray}
\nonumber
\frac{\partial \hat{\theta}}{\partial \gamma}|_{\gamma=0} &=&-
\begin{vmatrix}
\sigma_h^h &\sigma_l^h &0\\
\sigma_h^l &\sigma_l^l &0\\
\sigma_h^s &\sigma_l^s &-1
\end{vmatrix}
/|J_\sigma|=-(-\sigma_h^h \sigma_l^l+\sigma_h^l \sigma_l^h)/|J_\sigma|>0, \\ \nonumber
\frac{\partial \theta_h}{\partial \gamma}|_{\gamma=0} &=&-
\begin{vmatrix}
0 &\sigma_l^h &\sigma_s^h\\
0 &\sigma_l^l &\sigma_s^l\\
-1 &\sigma_l^s &\sigma_s^s
\end{vmatrix}
/|J_\sigma|=-(-\sigma_l^h \sigma_s^l+\sigma_l^l \sigma_s^h)/|J_\sigma|>0, \\ \nonumber
\frac{\partial \theta_l}{\partial \gamma}|_{\ \gamma=0} &=&-
\begin{vmatrix}
\sigma_h^h &0 &\sigma_s^h\\
\sigma_h^l &0 &\sigma_s^l\\
\sigma_h^s &-1 &\sigma_s^s
\end{vmatrix}
/|J_\sigma|=-(-\sigma_h^l \sigma_s^h+\sigma_h^h \sigma_s^l)/|J_\sigma|<0.
\end{eqnarray}

\end{proof}

\section*{References}
\begin{description}
\item Adriani, F. and Sonderegger, S. 2019. ``A theory of esteem based peer pressure.'' {\it Games and Economic Behavior}, 115, 314-335.
\item Babcock, Philip, and Mindy Marks. 2011. ``The Falling Time Cost of College : Evidence from Half a Century of Time Use Data." {\it Review of Economics and Statistics}, 93(2): 468-678.
\item Benabou, R. and J. Tirole. 2006. ``Incentives and Prosocial Behavior." {\it American Economic Review}, 96: 1652-1678.
\item Benabou, R. and J. Tirole. 2011. ``Laws and Norms." NBER Working Paper No. 17579.
\item Branscombe, N., Spears, R., Ellemers, N., and Doosje, B. 2002. ``Intragroup and Intergroup Evaluation Effects on Group Behavior." {\it Personality and Social Psychology Bulletin}, 28: 744-753.
\item Cho, I. K., \& and Kreps, D. M. (1987). Signaling games and stable equilibria. {\it Quarterly Journal of Economics}, 102, 179--221.
\item Dur, R. and J. van der Weele. 2012. ``Status-Seeking in Criminal Subcultures and the Double Dividend of Zero-Tolerance." {\it Journal of Public Economic Theory}, 15(1): 77-93.
\item E. Lazarova and D. Dimitrov. 2013. ``Status-seeking in hednic games with heterogeneous players." {\it Social Choice and Welfare}, 40: 1205-1229.
\item D. Kreps and R. Wilson. 1982. ``Sequential Equilibrium." {\it Econometrica}, 50: 863-894.
\item Hamada T. 2016. ``The Effect of College as a Signal on Sutudents' Activity." MIMEO. 
\item Heather J. Smith and Tom R. Tyler. 1996. ``Choosing the Right Pond: The Impact of Group Membership on Self-Esteem and Group-Oriented Behavior". {\it Journal of Experimental Social Psychology}, 33: 146-170.
\item Janabel, Jiger. 1996. ``When national ambition conflicts with reality: studies on Kazakhstan's ethnic relations." central Asian Survey 15(1), 5-21.
\item Jewitt, I. 2004. ``Notes on the Shape of Distributions." MIMEO, Oxford University.
\item MacLeod, W.B. and Urquiola, M. 2015. ``Reputation and school competition." {\it American Economic Review}, 105(11): 3471-3488.
\item Milchtaich, I. and Winter, E. 2002. ``Stability and Segregation in Group Formation". {\it Games and Economic Behavior}, 38: 318-346.
\item Sivanathan, Niro, Pettit, Nathan. C. 2010. ``Protecting the self through consumption: status goods as affirmational commodities". {\it Journal of Experimental Social Psychology}, 46, 564-570.
\item Tyler, Tom R. and E. Allan Lind. 1992. ``A relational model of authority in groups". {\it In Advances in Experimental Psychology}, 25: 115-191. New York: Academic Press.
\item Watts, A. 2007. ``Formation of segregated and integrated groups." {\it International Journal of Game Theory}, 35: 505-519.
\end{description}

\end{document}